\documentclass[10pt,journal,epsfig]{IEEEtran}

\usepackage[dvips]{graphicx}
\usepackage{graphicx}
\usepackage{amssymb}
\usepackage{cite}
\usepackage{subfigure}
\usepackage{amsmath}
\usepackage{algorithm}
\usepackage{algorithmic}
\usepackage{multirow}
\usepackage{caption}

\usepackage{tikz}
\IEEEoverridecommandlockouts

\newcommand\copyrighttext{%
  \footnotesize \textcopyright 2020 IEEE. Personal use of this material is permitted.
  Permission from IEEE must be obtained for all other uses, in any current or future
  media, including reprinting/republishing this material for advertising or promotional
  purposes, creating new collective works, for resale or redistribution to servers or
  lists, or reuse of any copyrighted component of this work in other works. DOI:10.1109/TVT.2020.3031657
  }
\newcommand\copyrightnotice{%
\begin{tikzpicture}[remember picture,overlay]
\node[anchor=south,yshift=8pt] at (current page.south) {\fbox{\parbox{\dimexpr\textwidth-\fboxsep-\fboxrule\relax}{\copyrighttext}}};
\end{tikzpicture}%
}

\begin{document}

\title{Intelligent Reflecting Surface-Assisted Millimeter Wave Communications:
Joint Active and Passive Precoding Design}

\author{Peilan Wang, Jun Fang, Xiaojun Yuan, Zhi Chen, and Hongbin Li, ~\IEEEmembership{Fellow,~IEEE}
\thanks{Peilan Wang, Jun Fang, Xiaojun Yuan and Zhi Chen are with the National Key Laboratory
of Science and Technology on Communications, University of
Electronic Science and Technology of China, Chengdu 611731, China,
Email: JunFang@uestc.edu.cn}
\thanks{Hongbin Li is
with the Department of Electrical and Computer Engineering,
Stevens Institute of Technology, Hoboken, NJ 07030, USA, E-mail:
Hongbin.Li@stevens.edu}
\thanks{This work was supported in part by the National Science
Foundation of China under Grant 61829103. }}

\maketitle

\copyrightnotice

\begin{abstract}
Millimeter wave (MmWave) communications is capable of supporting
multi-gigabit wireless access thanks to its abundant spectrum
resource. However, severe path loss and high directivity make it
vulnerable to blockage events, which can be frequent in indoor and
dense urban environments. To address this issue, in this paper, we
introduce intelligent reflecting surface (IRS) as a new technology
to provide effective reflected paths to enhance the coverage of
mmWave signals. In this framework, we study joint active and
passive precoding design for IRS-assisted mmWave systems, where
multiple IRSs are deployed to assist the data transmission from a
base station (BS) to a single-antenna receiver. Our objective is
to maximize the received signal power by jointly optimizing the
BS's transmit precoding vector and IRSs' phase shift coefficients.
Although such an optimization problem is generally non-convex, we
show that, by exploiting some important characteristics of mmWave
channels, an optimal closed-form solution can be derived for the
single IRS case and a near-optimal analytical solution can be
obtained for the multi-IRS case. Our analysis reveals that the
received signal power increases quadratically with the number of
reflecting elements for both the single IRS and multi-IRS cases.
Simulation results are included to verify the optimality and
near-optimality of our proposed solutions. Results also show that
IRSs can help create effective virtual line-of-sight (LOS) paths
and thus substantially improve robustness against blockages in
mmWave communications.
\end{abstract}


\section{Introduction}
Millimeter-wave (mmWave) communication is a promising technology
for future cellular networks
\cite{RappaportMurdock11,RanganRappaport14,GhoshThomas14}. It has
the potential to offer gigabits-per-second communication data
rates by exploiting the large bandwidth available at mmWave
frequencies. A key challenge for mmWave communication is that
signals experience a much more significant path loss over mmWave
frequency bands as compared with the path attenuation over lower
frequency bands \cite{SwindlehurstAyanoglu14}. To compensate for
the severe path loss in mmWave systems, large antenna arrays are
generally used to achieve significant beamforming gains for data
transmission \cite{AlkhateebMo14,LiMasouros16,ZhangGe16}. On the
other hand, high directivity makes mmWave communications
vulnerable to blockage, which can be frequent in indoor and dense
urban environments. For instance, due to the narrow beamwidth of
mmWave signals, a very small obstacle, such as a person's arm, can
effectively block the link \cite{AbariBharadia17}. To address this
issue, in some prior works, e.g.
\cite{YangDu15,ZubairJangsher19,NiuDing19}, relays are employed to
overcome blockage and improve the coverage of mmWave signals.

Recently, to address the blockage issue and enable indoor mobile
mmWave networks, reconfigurable reflect-arrays (also referred to
as intelligent reflecting surfaces) were introduced to establish
robust mmWave connections for indoor networks even when the
line-of-sight (LOS) link is blocked by obstructions, and the
proposed solution was validated by a test-bed with $14\times 16$
reflector units \cite{TanSun18}. Intelligent reflecting surface
(IRS) has been recently proposed as a promising new technology for
realizing a smart and programmable wireless propagation
environment via software-controlled reflection
\cite{CuiQi14,LiaskosNie18}. Specifically, IRS, made of a newly
developed metamaterial, is a planar array comprising a large
number of reconfigurable passive elements. With the aid of a smart
micro controller, each element can independently reflect the
incident signal with a reconfigurable amplitude and phase shift.
By properly adjusting the phase shifts of the passive elements,
the reflected signals can add coherently at the desired receiver
to improve the signal power or destructively at non-intended
receivers to suppress interference \cite{WuZhang19a}.

IRS-aided wireless communications have attracted much attention
recently \cite{WuZhang18,WuZhang19a,YanKuai19,YangZhang19,
HuangZappone18,HuangZappone19,NadeemKammoun19,HeYuan19,EmilLuca19}.
A key problem for IRS-aided systems is to jointly optimize the
active beamforming vector at the BS and the reflection
coefficients at the IRS to achieve different objectives. Such a
problem was studied in a single-user scenario, where the objective
was to maximize the receive signal power \cite{WuZhang18}. A
similar problem was considered in an orthogonal frequency division
multiplexing (OFDM)-based communication system \cite{YangZhang19},
with the objective of maximizing the achievable rate. In addition,
the joint BS-IRS optimization problem was investigated in a
downlink multi-user scenario, e.g.
\cite{HuangZappone18,HuangZappone19,NadeemKammoun19}. In
\cite{CuiZhang19,YuXu19,ShenXu19,GuanWu19,MishraJohansson19,WuZhang19b,LiBin19},
IRS was also considered as an auxiliary facility to assist secret
communications, unmanned aerial vehicle (UAV) communications and
wireless power transfer. In \cite{StefanRenzo19}, the joint
beamforming problem was studied to maximize the capacity of an
IRS-assisted MIMO indoor mmWave system.

Inspired by encouraging results reported in \cite{TanSun18}, in
this paper, we consider a scenario where multiple IRSs are
deployed to assist downlink point-to-point mmWave communications.
A joint active and passive precoding design problem is studied,
where the objective is to maximize the received signal power by
jointly optimizing the BS's transmit precoding vector and IRSs'
phase shift coefficients. Note that such a joint active and
passive precoding problem is non-convex and has been studied in
previous works \cite{WuZhang18,WuZhang19a} for conventional
microwave communication systems, where a single IRS is deployed to
assist the data transmission from the BS to the user. In
\cite{WuZhang18}, this non-convex problem was relaxed as a convex
semidefinite programming (SDP) problem. Nevertheless, the proposed
approach is sub-optimal and does not have an analytical solution.
In addition, solving the SDP problem usually involves a high
computational complexity.

In this paper, we will revisit this joint active and passive
beamforming problem from a mmWave communication perspective. We
show that, by exploiting some inherent characteristics of mmWave
channels, in particular an approximately rank-one structure of the
BS-IRS channel, an optimal closed-form solution can be derived for
the single IRS case and a near-optimal analytical solution can be
obtained for the multi-IRS case. Based on the analytical
solutions, we derive the maximum achievable average received
power, which helps gain insight into the effect of different
system parameters, including the number of passive reflecting
elements and the transmitter's antennas, on the system
performance. Our work focuses on a single data stream transmission
from the BS to the user. Although it is desirable to exploit
point-to-point multi-stream communications in order to improve the
spectral efficiency and achieve high data rates, there are still
some important scenarios where only a single-stream transmission
is available due to the rank-deficiency of the cascade channel
between the BS and the UE, or due to the use of a single
antenna/RF chain at the UE. In particular, mmWave has limited
diffraction and reflection abilities. Hence, multi-stream mmWave
communications may not be available for some indoor or outdoor
environments where the propagation is dominated by the LOS
component.

We noticed that the joint active and passive beamforming problem
for multi-user mmWave systems was studied in \cite{PradhanLi20},
where a sophisticated gradient-projection (GP) method was
developed. Nevertheless, due to the complex nature of the problem,
no analytical solution is available for the multi-user scenario.
In this case, the performance gain brought by the IRS has to be
demonstrated through numerical results as conducting a theoretical
analysis of the system performance is rather difficult.


The rest of the paper is organized as follows. In Section
\ref{sec:system-model}, the system model and the joint active and
passive precoding problem are discussed. The joint active and
passive precoding problem with a single IRS is studied in Section
\ref{sec:SingleIRS}, where a closed-form optimal solution is
developed and the average received power is analyzed. The joint
active and passive precoding problem with multiple IRSs is then
studied in Section \ref{sec:multi-IRS}, where a near-optimal
analytical solution is proposed. The extension of our proposed
solution to low-resolution phase shifters is discussed in Section
\ref{sec:low-resolution-PS}. Simulation results are presented in
Section \ref{sec:simulation-results}, followed by concluding
remarks in Section \ref{sec:conclusions}.

\section{System Model and Problem Formulation} \label{sec:system-model}
We consider an IRS-assisted mmWave downlink system as illustrated
in Fig.\ref{fig1}, where multiple IRSs are deployed to assist the
data transmission from the BS to a single-antenna user. Suppose
$K$ IRSs are employed to enhance the BS-user link, and the number
of reflecting units at each IRS is denoted by $M$. The BS is
equipped with $N$ antennas. Let ${\boldsymbol{h}}_d \in \mathbb
C^{N}$ denote the channel from the BS to the user,
$\boldsymbol{G}_k \in \mathbb C^{M\times N}$ denote the channel
from the BS to the $k$th IRS, and $\boldsymbol{h}_{{r_k}} \in
\mathbb C^{M}$ denote the channel from the $k$th IRS to the user.
Each element on the IRS behaves like a single physical point which
combines all the received signals and then re-scatters the
combined signal with a certain phase shift \cite{WuZhang18}. Let
$\theta_{k,m}\in [0,2\pi]$ denote the phase shift associated with
the $m$th passive element of the $k$th IRS. Define
\begin{align}
\boldsymbol{\Theta}_k\triangleq\text{diag}(
e^{j\theta_{k,1}},\ldots, e^{j\theta_{k,M}})
\end{align}
Let $\boldsymbol{w}\in\mathbb{C}^{N}$ denote the
precoding/beamforming vector used by the BS. The signal received
at the user can then be expressed as
\begin{align}
y = \left(\sum_{k=1}^K \boldsymbol{h}_{r_k}^H
\boldsymbol{\Theta}_k \boldsymbol{G}_k + \boldsymbol{h}_d^H
\right)\boldsymbol{w} s + \epsilon \label{received-signal-model}
\end{align}
where $s$ is the transmitted signal which is modeled as a random
variable with zero mean and unit variance, and $\epsilon$ denotes
the additive white Gaussian noise with zero mean and variance
$\sigma^2$. Note that in the above model, signals that are
reflected by the IRS two or more times are ignored due to the high
path loss of mmWave transmissions. Accordingly, the signal power
received at the user is given as
\begin{align}
\gamma=\left|\bigg(\sum_{k=1}^K \boldsymbol{h}_{r_k}^H
\boldsymbol{\Theta}_k \boldsymbol{G}_k + \boldsymbol{h}_d^H
\bigg)\boldsymbol{w}\right|^2
\end{align}

\begin{figure}[!t]
    \centering
    {\includegraphics[width=3.5in]{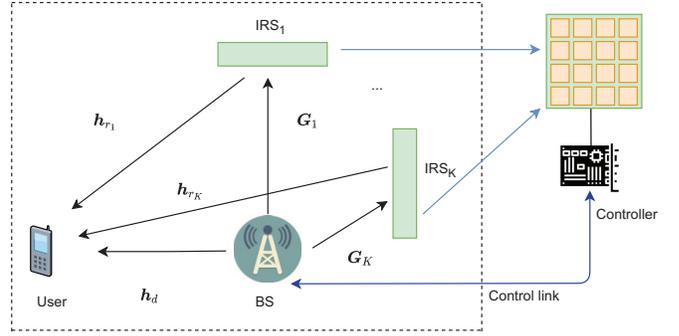}}
    \caption{IRS-assisted mmWave downlink system.} \label{fig1}
\end{figure}

In this paper, we assume that the knowledge of the global channel
state information is available. Channel estimation for
IRS-assisted systems can be found in, e.g.
\cite{DeepakHaakan19,WangFang19b,ChenLiang19,LiuGao20}. In particular,
\cite{WangFang19b,ChenLiang19,LiuGao20} discussed how to estimate the
channel for IRS-assisted mmWave systems. We aim to devise the
precoding vector $\boldsymbol{w}$ and the diagonal phase shift
matrices $\{\boldsymbol{\Theta}_k\}$ to maximize the received
signal power, i.e.
\begin{align}
\max_{\boldsymbol{w},\{\boldsymbol{\Theta}_k\}}\quad &
\left|\bigg(\sum_k^K \boldsymbol{h}_{r_k}^H \boldsymbol{\Theta}_k
\boldsymbol{G}_k + \boldsymbol{h}_d^H
\bigg)\boldsymbol{w}\right|^2 \nonumber\\
\text{s.t.} \quad & \|\boldsymbol{w}\|_2^2\leq p \nonumber\\
& \boldsymbol{\Theta}_k=\text{diag}( e^{j\theta_{k,1}},\ldots,
e^{j\theta_{k,M}}) \quad \forall k \label{opt1}
\end{align}
where $p$ denotes the maximum transmit power at the BS. Note that
here we only consider the communication power. In practical
systems, the computation cost may also need to be considered in
order to achieve a good balance between the energy efficiency and
the spectral efficiency \cite{GeSun18}. Such a tradeoff for
IRS-aided mmWave systems is an interesting and important topic
worthy of future investigation. The problem (\ref{opt1}) is
referred to as joint active and passive beamforming. Note that the
optimization problem (\ref{opt1}) with $K=1$ has been studied in
\cite{WuZhang18}, where the nonconvex problem was relaxed as a
convex semidefinite programming (SDP) problem. Nevertheless, the
proposed approach is generally sub-optimal and does not have an
analytical solution. Besides, solving the SDP problem involves a
high computational complexity.

In this paper, we will revisit this joint active and passive
beamforming problem for mmWave communications by exploiting some
inherent characteristics of mmWave channels. Specifically, mmWave
channels are typically sparsely-scattered. A widely used
Saleh-Valenzuela (SV) channel model for mmWave communications is
given as \cite{AyachRajagopal14,GaoDai16,AkdenizLiu14}:
\begin{align}
\boldsymbol{H}=\sqrt{\frac{N_t
N_r}{L}}\bigg(\beta_0\boldsymbol{a}_{r}(\varphi_{0}^{r})
\boldsymbol{a}_{t}^{H}(\varphi_{0}^{t})+\sum_{i=1}^{L-1}\beta_i\boldsymbol{a}_{r}(\varphi_{i}^{r})
\boldsymbol{a}_{t}^{H}(\varphi_{i}^{t})\bigg)
\end{align}
where $N_t$ and $N_r$ respectively denote the number of antennas
at the transmitter and the receiver, $L$ is the total number of
paths, $\beta_0\boldsymbol{a}_{r}(\varphi_{0}^{r})
\boldsymbol{a}_{t}(\varphi_{0}^{t})$ is the LOS component with
$\beta_0$ representing the complex gain, $\varphi_{0}^{r}$
representing the angle of arrival at the receiver, and
$\varphi_{0}^{t}$ representing the angle of departure at the
transmitter, and $\beta_i\boldsymbol{a}_{r}(\varphi_{i}^{r})
\boldsymbol{a}_{t}^{H}(\varphi_{i}^{t})$ denotes the $i$th
non-line-of-sight (NLOS) component. Also,
$\boldsymbol{a}_{r}(\varphi_{i}^{r})$ and
$\boldsymbol{a}_{t}(\varphi_{i}^{t})$ denote the array response
vectors associated with the receiver and the transmitter,
respectively. In addition to the sparse scattering
characteristics, many measurement campaigns reveal that the power
of the mmWave LOS path is much higher (about 13dB higher) than the
sum of power of NLOS paths
\cite{Muhi-EldeenIvrissimtzis10,SamimiMacCartney16}. Hence, any
system that is not centered around the transmission via the direct
LOS path usually gives only limited gains \cite{StefanRenzo19}.
Motivated by this fact, it is highly desirable to ensure that the
channel between the BS and each IRS is LOS dominated. In practice,
with the knowledge of the location of the BS, IRSs can be
installed within sight of the BS. Since the power of NLOS paths is
negligible compared to that of the LOS path, the BS-IRS channel
can be well approximated as a rank-one matrix, i.e.
\begin{align}
\boldsymbol{G}_k\approx\lambda_k\boldsymbol{a}_k\boldsymbol{b}_k^T
\quad \forall k \label{Gk}
\end{align}
where $\lambda_k$ is a scaling factor accounting for antenna and
path gains, $\boldsymbol{a}_k\in\mathbb{C}^{M}$ and
$\boldsymbol{b}_k\in\mathbb{C}^{N}$ represent the normalized array
response vector associated with the IRS and the BS, respectively.
As will be shown later in this paper, this rank-one channel
structure can be utilized to obtain a closed-form solution to
(\ref{opt1}). Also, our simulation results show that our proposed
solution based on this rank-one approximation can achieve a
received signal power that is nearly the same as that attained by
taking those NLOS paths into account.


\section{Joint Precoding Design for Single IRS}\label{sec:SingleIRS}
\subsection{Optimal Solution}
In this section, we first consider the case where there is only a
single IRS, i.e. $K=1$. We omit the subscript $k$ for simplicity
in the single IRS case. The optimization (\ref{opt1}) is
simplified as
\begin{align}
\max_{\boldsymbol{w},\boldsymbol{\Theta}}\quad & \left|\left(
\boldsymbol{h}_{r}^H \boldsymbol{\Theta} \boldsymbol{G} +
\boldsymbol{h}_d^H
\right)\boldsymbol{w}\right|^2 \nonumber\\
\text{s.t.} \quad & \|\boldsymbol{w}\|_2^2\leq p \nonumber\\
& \boldsymbol{\Theta}=\text{diag}( e^{j\theta_{1}},\ldots,
e^{j\theta_{M}}) \label{opt2}
\end{align}
We will show that by exploiting the rank-one structure of the
BS-IRS channel matrix $\boldsymbol{G}$, a closed-form solution to
(\ref{opt2}) can be obtained. Substituting
$\boldsymbol{G}=\lambda\boldsymbol{a}\boldsymbol{b}^T$ into the
objective function of (\ref{opt2}), we obtain
\begin{align}
|( \boldsymbol{h}_{r}^H \boldsymbol{\Theta} \boldsymbol{G} +
\boldsymbol{h}_d^H )\boldsymbol{w}|^2 =&
|\lambda\boldsymbol{h}_{r}^H\boldsymbol{\Theta}\boldsymbol{a}\boldsymbol{b}^T\boldsymbol{w}
+\boldsymbol{h}_d^H\boldsymbol{w}|^2 \nonumber\\
\stackrel{(a)}{=}&
|\eta\boldsymbol{\theta}^T\boldsymbol{g}+\boldsymbol{h}_d^H\boldsymbol{w}|^2
\nonumber\\
\stackrel{(b)}{=}&|\eta\boldsymbol{\bar{\theta}}^T\boldsymbol{g}e^{j\alpha}+\boldsymbol{h}_d^H\boldsymbol{w}|^2
\nonumber\\
\stackrel{(c)}{\leq}&|\eta\boldsymbol{\bar{\theta}}^T\boldsymbol{g}|^2+|\boldsymbol{h}_d^H\boldsymbol{w}|^2+
2|\eta\boldsymbol{\bar{\theta}}^T\boldsymbol{g}|\cdot|\boldsymbol{h}_d^H\boldsymbol{w}|
\end{align}
where in $(a)$, we define
$\eta\triangleq\boldsymbol{b}^T\boldsymbol{w}$,
$\boldsymbol{g}\triangleq
\lambda(\boldsymbol{h}_r^{\ast}\circ\boldsymbol{a})$, $\circ$
denotes the Hadamard (elementwise) product, and
\begin{align}
\boldsymbol{\theta}\triangleq
[e^{j\theta_{1}}\phantom{0}\ldots\phantom{0} e^{j\theta_{M}}]^T
\end{align}
in $(b)$, we write
$\boldsymbol{\theta}=\boldsymbol{\bar{\theta}}e^{j\alpha}$, and
the inequality $(c)$ becomes an equality when the arguments (also
referred to as phases) of the two complex numbers
$\eta\boldsymbol{\bar{\theta}}^T\boldsymbol{g}e^{j\alpha}$ and
$\boldsymbol{h}_d^H\boldsymbol{w}$ are identical. It should be
noted that we can always find an $\alpha$ such that the arguments
of $\beta\boldsymbol{\bar{\theta}}^T\boldsymbol{g}e^{j\alpha}$ and
$\boldsymbol{h}_d^H\boldsymbol{w}$ are identical, although at this
point we do not know the exact value of $\alpha$. Therefore the
optimization (\ref{opt2}) can be rewritten as
\begin{align}
\max_{\boldsymbol{w},\boldsymbol{\bar{\theta}}}\quad &
|\eta\boldsymbol{\bar{\theta}}^T\boldsymbol{g}|^2+|\boldsymbol{h}_d^H\boldsymbol{w}|^2+
2|\eta\boldsymbol{\bar{\theta}}^T\boldsymbol{g}|\cdot|\boldsymbol{h}_d^H\boldsymbol{w}|
 \nonumber\\
\text{s.t.} \quad & \|\boldsymbol{w}\|_2^2\leq p  \label{opt3}
\end{align}
It is clear that the optimization of $\boldsymbol{\bar{\theta}}$
is independent of $\boldsymbol{w}$, and
$\boldsymbol{\bar{\theta}}$ can be solved via
\begin{align}
\max_{\boldsymbol{\bar{\theta}}}\quad &
|\boldsymbol{\bar{\theta}}^T\boldsymbol{g}| \nonumber\\
\text{s.t.} \quad &
\boldsymbol{\bar{\theta}}=[e^{j\bar{\theta}_{1}}\phantom{0}\ldots\phantom{0}
e^{j\bar{\theta}_{M}}]^T
\end{align}
It can be easily verified that the objective function reaches its
maximum $\| \boldsymbol g\|_{1}$ when
\begin{align}
\boldsymbol{\bar{\theta}}^{\star}=[e^{-j\text{arg}(g_1)}\phantom{0}\ldots\phantom{0}e^{-j\text{arg}(g_M)}]^T
\end{align}
where $\text{arg}(x)$ denotes the argument of the complex number
$x$, and $g_m$ denotes the $m$th entry of $\boldsymbol{g}$.

So far we have obtained the optimal solution of
$\boldsymbol{\bar{\theta}}$, which, as analyzed above, is
independent of the optimization variables $\alpha$ and
$\boldsymbol{w}$. Based on this result, the optimization
(\ref{opt2}) can be simplified as
\begin{align}
\max_{\boldsymbol{w},\alpha}\quad & \left|\left(
e^{j\alpha}\boldsymbol{h}_{r}^H
\boldsymbol{\bar{\Theta}}^{\star}\boldsymbol{G} +
\boldsymbol{h}_d^H
\right)\boldsymbol{w}\right|^2 \nonumber\\
\text{s.t.} \quad & \|\boldsymbol{w}\|_2^2\leq p \label{opt4}
\end{align}
where
$\boldsymbol{\bar{\Theta}}^{\star}\triangleq\text{diag}(\boldsymbol{\bar{\theta}}^{\star})$.
For a fixed $\alpha$, it is clear that the optimal precoding
vector $\boldsymbol{w}$, also known as the maximum ratio
transmission (MRT) solution, is given by
\begin{align}
\boldsymbol{w}^{\star}=\sqrt{p}\frac{\left(
e^{j\alpha}\boldsymbol{h}_{r}^H
\boldsymbol{\bar{\Theta}}^{\star}\boldsymbol{G} +
\boldsymbol{h}_d^H \right)^H}{\|e^{j\alpha}\boldsymbol{h}_{r}^H
\boldsymbol{\bar{\Theta}}^{\star}\boldsymbol{G} +
\boldsymbol{h}_d^H\|_2} \label{eqn1}
\end{align}
By substituting the optimal precoding vector
$\boldsymbol{w}^{\ast}$ into (\ref{opt4}), the problem becomes
optimization of $\alpha$:
\begin{align}
\max_{\alpha}\quad \| e^{j\alpha}\boldsymbol{h}_{r}^H
\boldsymbol{\bar{\Theta}}^{\star}\boldsymbol{G} +
\boldsymbol{h}_d^H \|_2^2 \label{opt5}
\end{align}
whose optimal solution can be easily obtained as
\begin{align}
\alpha^{\star}= & -\text{arg}\left((\boldsymbol{h}_{r}^H
\boldsymbol{\bar{\Theta}}^{\star}\boldsymbol{G})\boldsymbol{h}_d\right)
\nonumber\\
=& -\text{arg}\left((\lambda\boldsymbol{h}_{r}^H
\boldsymbol{\bar{\Theta}}^{\star}\boldsymbol{a}\boldsymbol{b}^T)
\boldsymbol{h}_d\right) \nonumber\\
=&-\text{arg}\left(\boldsymbol {b}^T \boldsymbol{h}_d\right)
\label{alpha-optimal}
\end{align}
where the last equality follows from the fact that
$\lambda\boldsymbol{h}_{r}^H\boldsymbol{\bar{\Theta}}^{\star}\boldsymbol{a}=\boldsymbol{g}^T
\boldsymbol{\bar{\theta}}^{\star}=\|\boldsymbol{g}\|_1$ is a
real-valued number. After the optimal value of $\alpha$ is
obtained, the optimal precoding vector can be determined by
substituting (\ref{alpha-optimal}) into (\ref{eqn1}), and the
optimal diagonal phase shift matrix is given as
\begin{align}
\boldsymbol{\Theta}^{\star}=e^{j\alpha^{\star}}\boldsymbol{\bar{\Theta}}^{\star}
\end{align}
We see that under the rank-one BS-IRS channel assumption, a
closed-form solution to the joint active and passive beamforming
problem (\ref{opt2}) can be derived. To calculate this optimal
solution, we only need to compute $\boldsymbol{b}^T
\boldsymbol{h}_d$ and $\boldsymbol{g}$, which involves a
computational complexity of $\mathcal{O}(\max(M,N))$.

\subsection{Power Scaling Law}
We now characterize the scaling law of the average received power
with respect to the number of reflecting elements $M$. For
simplicity, we set $p=1$. Our main results are summarized as
follows.

\newtheorem{proposition}{Proposition}
\begin{proposition}
Assume $\boldsymbol{h}_{r}\sim {\cal
CN}(0,\varrho_{r}^2\boldsymbol{I})$, $\boldsymbol{h}_d\sim {\cal
CN}(0, \varrho_d^2\boldsymbol{I})$, and the BS-IRS channel is
characterized by a rank-one geometric model given as
\begin{align}
\boldsymbol{G}=\sqrt{N M}\rho \boldsymbol{a}\boldsymbol{b}^T
\end{align}
where $\rho$ denotes the complex gain associated with the LOS path
between the BS and the IRS, $\boldsymbol{a}\in\mathbb{C}^{M}$ and
$\boldsymbol{b}\in\mathbb{C}^{N}$ are normalized array response
vectors associated with the IRS and the BS, respectively. Then the
average received power at the user attained by the optimal
solution of (\ref{opt2}) is given as
\begin{align}
\gamma^{\star}=&N M^2 \frac{{\pi}\varrho_{r}^2}{4}
\mathbb{E}[|\rho|^2]+2M\sqrt{N}\mathbb E[|\rho|]\frac{{\pi}
\varrho_r\varrho_d}{4}
\nonumber\\&   + N M\left(2-\frac{\pi}{2} \right) \mathbb E[|\rho|^2] \frac{\varrho_{r}^2}{2}+
N\varrho_d^2 \label{power-scaling-law}
\end{align}
 \label{proposition1}
\end{proposition}
\begin{proof}
See Appendix \ref{appA}.
\end{proof}

From (\ref{power-scaling-law}), we see that the average received
signal power attained by the optimal beamforming solution scales
quadratically with the number of reflecting elements $M$. Such a
``squared improvement'' is due to the fact that the optimal
beamforming solution not only allows to achieve a transmit
beamforming gain of $M$ in the IRS-user link, but it also acquires
a gain of $M$ by coherently collecting signals in the BS-IRS link.
This result implies that scaling up the number of reflecting
elements is a promising way to compensate for the significant path
loss in mmWave wireless communications.

\section{Joint Precoding Design for Multiple IRSs} \label{sec:multi-IRS}
In this section, we return to the joint active and passive
beamforming problem (\ref{opt1}) for the general multi-IRS setup.
Such a problem is more challenging as we need to jointly design
the precoding vector $\boldsymbol{w}$ and a set of phase shift
matrices associated with $K$ IRSs. In the following, by exploiting
the rank-one structure of BS-IRS channels and the
near-orthogonality between array response vectors, we show that a
near-optimal analytical solution can be obtained for this
nonconvex problem.

\subsection{Proposed Solution}
Substituting
$\boldsymbol{G}_k=\lambda_k\boldsymbol{a}_k\boldsymbol{b}_k^T$
into the objective function of (\ref{opt1}), we arrive at
\begin{align}
\quad &\left|\bigg(\sum_{k=1}^K \boldsymbol{h}_{r_k}^H
\boldsymbol{\Theta}_k \boldsymbol{G}_k + \boldsymbol{h}_d^H
\bigg)\boldsymbol{w}\right|^2 \nonumber\\= &
\left|\bigg(\sum_{k=1}^K \lambda_k\boldsymbol{h}_{r_k}^H
\boldsymbol{\Theta}_k \boldsymbol{a}_k\boldsymbol{b}_k^T  +
\boldsymbol{h}_d^H \bigg)\boldsymbol{w}\right|^2 \nonumber
\\
\stackrel{(a)}{=}& \left|
\sum_{k=1}^K\eta_k\boldsymbol{\theta}_k^T\boldsymbol{g}_k+\boldsymbol{h}_d^H\boldsymbol{w}\right|^2
\nonumber\\
\stackrel{(b)}{=}&\left|
\sum_{k=1}^K\eta_k\boldsymbol{\bar{\theta}}_k^T\boldsymbol{g}_k
e^{j\alpha_k}+\boldsymbol{h}_d^H\boldsymbol{w}\right|^2
\nonumber\\
\stackrel{(c)}{\leq}& \sum_{k=1}^K
\left|\eta_k\boldsymbol{\bar{\theta}}_k^T\boldsymbol{g}_k\right|^2
+\sum_{i=1}^K \sum_{j\neq i}^K
|\eta_i\boldsymbol{\bar{\theta}}_i^T
\boldsymbol{g}_i|\cdot|\eta_j\boldsymbol{\bar{\theta}}_j^T\boldsymbol{g}_j| \nonumber \\
&+|\boldsymbol{h}_d^H\boldsymbol{w}|^2+
2\sum_{k=1}^K|\eta_k\boldsymbol{\bar{\theta}}_k^T\boldsymbol{g}_k|\cdot|\boldsymbol{h}_d^H\boldsymbol{w}|
\label{upper-bound}
\end{align}
where in $(a)$, we define
$\eta_k\triangleq\boldsymbol{b}_k^T\boldsymbol{w}$,
$\boldsymbol{g}_k\triangleq
\lambda_k(\boldsymbol{h}_{r_k}^{\ast}\circ\boldsymbol{a}_k)$, and
$\boldsymbol{\theta}_k\triangleq
[e^{j\theta_{k,1}}\phantom{0}\ldots\phantom{0}
e^{j\theta_{k,M}}]^T$, in $(b)$, we write
$\boldsymbol{\theta}_k=\boldsymbol{\bar{\theta}}_k e^{j\alpha_k}$,
and the inequality $(c)$ becomes an equality when the arguments
(or phases) of all complex numbers inside the brackets of $(b)$
are identical. It should be noted that there exist a set of
$\{\alpha_k\}$ such that the arguments of
$\eta_k\boldsymbol{\bar{\theta_k}}^T\boldsymbol{g}_k
e^{j\alpha_k},\forall k$ and $\boldsymbol{h}_d^H\boldsymbol{w}$
are identical, although at this point we do not know the values of
$\{\alpha_k\}$. Therefore \eqref{opt1} is equivalent to maximizing
the upper bound given in (\ref{upper-bound}), i.e.
\begin{align}
\max_{\boldsymbol{w},\{\boldsymbol{\bar{\theta}}_k\}}\quad &
\sum_{k=1}^K
\left|\eta_k\boldsymbol{\bar{\theta}}_k^T\boldsymbol{g}_k\right|^2+\sum_{i=1}^K
\sum_{j\neq i}^K
|\eta_i\boldsymbol{\bar{\theta}}_i^T\boldsymbol{g}_i|\cdot|\eta_j\boldsymbol{\bar{\theta}}_j^T\boldsymbol{g}_j|
\nonumber\\&+|\boldsymbol{h}_d^H\boldsymbol{w}|^2 +
2\sum_{k=1}^K|\eta_k\boldsymbol{\bar{\theta}}_k^T\boldsymbol{g}_k|\cdot|\boldsymbol{h}_d^H\boldsymbol{w}|
 \nonumber\\
\text{s.t.} \quad & \|\boldsymbol{w}\|_2^2\leq p \label{opt6}
\end{align}
From (\ref{opt6}), it is clear that the optimization of
$\{\boldsymbol{\bar{\theta}}_k\}$ can be decomposed into a number
of independent sub-problems, with $\boldsymbol{\bar{\theta}}_k$
solved by
\begin{align}
\max_{\boldsymbol{\bar{\theta}}_k}\quad &
|\boldsymbol{\bar{\theta}}_k^T\boldsymbol{g}_k| \nonumber\\
\text{s.t.} \quad &
\boldsymbol{\bar{\theta}}_k=[e^{j\bar{\theta}_{k,1}}\phantom{0}\ldots\phantom{0}
e^{j\bar{\theta}_{k,M}}]^T
\end{align}
It can be easily verified that the objective function reaches its
maximum $\| \boldsymbol{g}_k\|_{1}$ when
\begin{align}
\boldsymbol{\bar{\theta}}_k^{\star} = [ e^{-j {\rm arg} (g_{k,1})}
\phantom{0}\ldots\phantom{0} e^{-j {\rm arg} (g_{k,M})} ]
\end{align}
where $g_{k,m}$ denotes the $m$th entry of $\boldsymbol{g}_k$.

So far we have obtained the optimal solution of
$\{\boldsymbol{\bar{\theta}}_k\}$, which, as analyzed above, is
independent of the optimization variables $\{\alpha_k\}$ and
$\boldsymbol{w}$. Based on this result, the optimization
(\ref{opt1}) can be reformulated as
\begin{align}
\max_{\boldsymbol{w},\{\alpha_k\}}\quad & \left|\bigg(\sum_{k=1}^K
\lambda_k e^{j\alpha_k}\boldsymbol{h}_{r_k}^H
\boldsymbol{\bar{\Theta}}_k^{\star}
\boldsymbol{a}_k\boldsymbol{b}_k^T + \boldsymbol{h}_d^H
\bigg)\boldsymbol{w}\right|^2 \nonumber\\
\text{s.t.} \quad & \|\boldsymbol{w}\|_2^2\leq p \label{opt7}
\end{align}
where
$\boldsymbol{\bar{\Theta}}_k^{\star}\triangleq\text{diag}(\boldsymbol{\bar{\theta}}_k^{\star})$.
Note that
\begin{align}
\lambda_k\boldsymbol{h}_{r_k}^H
\boldsymbol{\bar{\Theta}}_k^{\star}\boldsymbol{a}_k =
\boldsymbol{g}_k^T
\boldsymbol{\bar{\theta}}_k^{\star}=\|\boldsymbol{g}_k\|_1\triangleq
z_k
\end{align}
is a real-valued number. Thus the objective function of
(\ref{opt7}) can be written in a more compact form as
\begin{align}
\left|\bigg(\sum_{k=1}^K z_k e^{j\alpha_k} \boldsymbol{b}_{k}^T+
\boldsymbol{h}_d^H \bigg)\boldsymbol{w}\right|^2
\stackrel{(a)}{=}&\left| \bigg(\boldsymbol{v}^H
\boldsymbol{D}_z\boldsymbol{B} + \boldsymbol{h}_{d}^H \bigg) \boldsymbol{w}\right|^2 \nonumber \\
\stackrel{(b)}{=}& \left| \bigg(\boldsymbol{v}^H \boldsymbol{
\Phi} + \boldsymbol{h}_{d}^H \bigg) \boldsymbol{w}\right|^2
\end{align}
where in $(a)$, we define $\boldsymbol{v}\triangleq
[e^{j\alpha_1}\phantom{0}\ldots\phantom{0} e^{j\alpha_K}]^H$,
$\boldsymbol{D}_z\triangleq \text{diag}(z_1,\ldots, z_K)$ and
$\boldsymbol{B}\triangleq[\boldsymbol{b}_1\phantom{0}\ldots\phantom{0}
\boldsymbol{b}_K]^T $, and in $(b)$, we define
$\boldsymbol{\Phi}\triangleq\boldsymbol{D}_z\boldsymbol{B}$. Hence
(\ref{opt7}) can be simplified as
\begin{align}
\max_{\boldsymbol{w},{\boldsymbol{v}}}\quad &
\left| \bigg(\boldsymbol{v}^H  \boldsymbol{ \Phi} + \boldsymbol{h}_{d}^H \bigg) \boldsymbol{w}\right|^2 \nonumber\\
\text{s.t.} \quad & \|\boldsymbol{w}\|_2^2\leq p \label{opt8}
\end{align}
Note that for any given $\boldsymbol{v}$, an optimal precoding
vector $\boldsymbol{w}$ , i.e. the MRT solution, is given as
\begin{align}
\boldsymbol{w}^{\star} = \sqrt{p}\frac{\left(\boldsymbol{v}^H
\boldsymbol{\Phi}+\boldsymbol{h}_{d}^H\right)^H}{\|
\boldsymbol{v}^H\boldsymbol{\Phi}+\boldsymbol{h}_{d}^H\|_2}
\label{MRT-solution}
\end{align}
Substituting the optimal precoding vector $\boldsymbol{w}^{\star}$
into the objective function of (\ref{opt8}) yields
\begin{align}
\max_{{\boldsymbol{v}}}\quad &
\| \boldsymbol{v}^H  \boldsymbol{ \Phi} + \boldsymbol{h}_d^H\|_2^2 \nonumber\\
\text{s.t.} \quad & \boldsymbol{v} =
[e^{j\alpha_1}\phantom{0}\ldots\phantom{0}  e^{j\alpha_K}]^H
\label{opt12}
\end{align}
or equivalently,
\begin{align}
 \max_{\boldsymbol{v}} \quad & \boldsymbol{v}^H \boldsymbol{\Phi \Phi}^H
 \boldsymbol{v}+\boldsymbol{v}^H \boldsymbol{\Phi} \boldsymbol h_d +
 \boldsymbol h_d^H\boldsymbol{\Phi} ^H \boldsymbol{v} \nonumber \\
\text{s.t.} \quad & |v_k| = 1 \quad \forall k \label{opt9}
\end{align}
Due to the unit circle constraint placed on entries of
$\boldsymbol{v}$, the above optimization (\ref{opt9}) is
non-convex. In the following, we first develop a sub-optimal
semidefinite relaxation (SDR)-based method to solve (\ref{opt9}).
Then, we show that by utilizing the near-orthogonality among array
response vectors, a near-optimal analytical solution of
(\ref{opt9}) can be obtained.

\subsubsection{A SDR-Based Approach for Solving (\ref{opt9})} Note
that (\ref{opt9}) is a non-convex quadratically constrained
quadratic program (QCQP), which can be reformulated as a
homogeneous QCQP by introducing an auxiliary variable $t$:
\begin{align}
\max_{\boldsymbol{\bar{v}}} \quad &\boldsymbol{\bar{v}}^H
\boldsymbol{R} \boldsymbol{\bar{v}} \nonumber
\\
\text{s.t.}\quad &|\bar{v}_k| = 1\quad \forall
k\in\{1,\ldots,K+1\} \label{opt10}
\end{align}
where
\begin{align*}
    \boldsymbol{R}\triangleq\left[\begin{matrix}
        \boldsymbol{\Phi \Phi}^H & \boldsymbol{\Phi} \boldsymbol h_d\\
        \boldsymbol h_d^H \boldsymbol \Phi^H & 0
    \end{matrix}
    \right],
\quad    \boldsymbol{\bar{v}}\triangleq \left[\begin{matrix}
        \boldsymbol v\\
        t
    \end{matrix}\right]
\end{align*}
and $\bar{v}_k$ denotes the $k$th entry of $\boldsymbol{\bar{v}}$.
Note that $\boldsymbol{\bar{v}}^H \boldsymbol{R}
\boldsymbol{\bar{v}}=\text{tr}(\boldsymbol{R}\boldsymbol{V})$,
where $\boldsymbol{V}\triangleq\boldsymbol{\bar{v}}
\boldsymbol{\bar{v}}^H$ is a rank-one and positive semidefinite
matrix, i.e. ${\boldsymbol V \succcurlyeq} 0$. Relaxing the
rank-one constraint, the problem \eqref{opt10} becomes
\begin{align}
\max_{\boldsymbol{V}}  \quad &\text{tr}(\boldsymbol R \boldsymbol
V) \nonumber
\\
\text{s.t.} \quad & \boldsymbol{V}_{k,k} = 1 \quad \forall k\nonumber \\
\quad &{\boldsymbol V} \succcurlyeq 0 \label{opt11}
\end{align}
where $\boldsymbol{V}_{k,k}$ denotes the $k$th diagonal element of
$\boldsymbol{V}$. The problem above is a standard convex
semidefinite program (SDP) which can be solved by convex tools
such as CVX. It can be readily verified that the computational
complexity for solving (\ref{opt11}) is at the order of
$\mathcal{O}((K+1)^6)$. In general, the optimal solution of
\eqref{opt11} is not guaranteed to be a rank-one matrix. To obtain
a rank-one solution from the obtained higher-rank solution of
\eqref{opt11}, one can follow the steps described in
\cite{SoZhang07}.

\subsubsection{Near-Optimal Analytical Solution To (\ref{opt9})}
The SDR-based method discussed above does not yield a closed-form
solution and is computationally expensive. In the following, we
propose a near-optimal analytical solution to (\ref{opt9}) via
utilizing the near-orthogonality among different steering vectors
$\{\boldsymbol{b}_k\}$.

Suppose a uniform linear array is employed at the BS. It can be
easily verified that the inner product of the two distinct array
response vectors $\boldsymbol{b}_i$ and $\boldsymbol{b}_j$ is
given as
\begin{align}
    \boldsymbol{b}_i^H \boldsymbol{b}_j = \frac{1}{N} \frac{1-e^{jN\delta} }{ 1-e^{j \delta}}
\end{align}
where
\begin{align}
\delta \triangleq\frac{2 \pi d}{\lambda} \left(\sin(\phi_i)-
\sin(\phi_j) \right)
\end{align}
in which $d$ denotes the distance between neighboring antenna
elements, $\lambda$ is the signal wavelength, and $\phi_i$ denotes
the angle of departure associated with the array response vector
$\boldsymbol{b}_i$. It is clear that
\begin{align}
|\boldsymbol{b}_i^H \boldsymbol{b}_j | \rightarrow 0, \quad
\text{as}\quad N\rightarrow \infty
\end{align}
In \cite{Chen13}, it was shown that asymptotic orthogonality still
holds for uniform rectangular arrays. Due to the small wavelength
at the mmWave frequencies, the antenna size is very small, which
allows a large number (hundreds or thousands) of array elements to
be packed into a small area in practical systems. In addition, to
improve the coverage, it is expected that different IRSs should be
deployed such that they, as seen from the BS, are sufficiently
separated in the angular domain, i.e. the angles of departure
$\{\phi_k\}$ are sufficiently separated. Taking into account these
factors, it is reasonable to assume that different steering
vectors $\{\boldsymbol{b}_k\}$ are near-orthogonal to each other,
i.e. $|\boldsymbol{b}_i^H\boldsymbol{b}_j|\approx 0$. Therefore we
have
\begin{align}
\boldsymbol v^H \boldsymbol{\Phi \Phi}^H
\boldsymbol v &= \boldsymbol v^H \boldsymbol{D}_z
\boldsymbol{B}\boldsymbol{B}^H\boldsymbol{D}_z\boldsymbol{v} \nonumber \\
& \approx \sum_{k=1}^K z_k^2 =\| \boldsymbol z\|_2^2 \label{eqn8}
\end{align}
which is a constant independent of the vector $\boldsymbol{v}$.
Consequently, the optimization (\ref{opt9}) can be simplified as
\begin{align}
\max_{\boldsymbol v} \quad & \boldsymbol{v}^H \boldsymbol{\Phi} \boldsymbol h_d +
\boldsymbol h_d^H\boldsymbol{\Phi} ^H \boldsymbol {v}\nonumber \\
\text{s.t.}  \quad &|v_k| = 1 \quad \forall k \label{opt13}
\end{align}
It can be easily verified that the optimal solution to
(\ref{opt13}) is given by
\begin{align}
\boldsymbol{v} ^{\star} = [ e^{-j {\rm arg} (u_{1})}
\phantom{0}\ldots\phantom{0} e^{-j {\rm arg} (u_{K})} ]^H
\label{v-opt}
\end{align}
where $\boldsymbol{u}\triangleq\boldsymbol{\Phi} \boldsymbol h_d$,
and $u_{k}$ denotes the $k$th entry of $\boldsymbol{u}$. After the
near-optimal phase vector $\boldsymbol{v}$ is obtained, it can be
substituted into (\ref{MRT-solution}) to obtain the precoding
vector $\boldsymbol{w}$. Also, the near-optimal diagonal phase
shift matrix associated with the $k$th IRS is given by
\begin{align}
\boldsymbol{\Theta}_k^{\star}=\boldsymbol{\bar{\Theta}}_k^{\star}
e^{j\alpha_k^{\star}} \label{opt-theta-m}
\end{align}
where $e^{j\alpha_k^{\star}}$ is the $k$th entry of
$\boldsymbol{v}^{\star}$. To calculate this near-optimal
analytical solution, the dominant operation includes calculating
$\boldsymbol{u} = \boldsymbol{\Phi} \boldsymbol{h}_d$ and
$\boldsymbol{g}_k = \lambda_k ( \boldsymbol{h}_{r_k}^{\ast} \circ
\boldsymbol{a}_k)$, which has a computational complexity of the
order $\mathcal{O}(\max(KN,M))$.

\subsection{Power Scaling Law} We now analyze the scaling law of
the average received power in the general multi-IRS setup with
respect to the number of passive elements $M$. Again, we set $p=1$
for simplicity. Our main results are summarized as follows.
\begin{proposition}
Assume $\boldsymbol {h}_{r_k}\sim {\cal CN}(0, \varrho_{r_k}^2
I)$, $\boldsymbol h_d \sim {\cal CN}(0, \varrho_d^2 I)$, and the
BS-IRS channel is characterized by a rank-one geometric model
given as
\begin{align}
\pmb G_k =  \sqrt{N M}\rho_k \boldsymbol{a}_k \boldsymbol{b}_k^T
\end{align}
where $\rho_{k}$ denotes the complex gain associated with the LOS
path between the BS and the $k$th IRS, $\boldsymbol{a}_k \in
\mathbb{C}^{M}$ and $\boldsymbol{b}_k\in\mathbb{C}^{N}$ are
normalized array response vectors associated with the IRS and the
BS, respectively. Then the average received power attained by the
near-optimal analytical solution is given by
\begin{align}
\gamma  \approx & N M^2
\sum_{k=1}^K\bigg(\frac{{\pi}\varrho_{r_k}^2} {4}
\mathbb{E}[|\rho_k|^2]\bigg)+ 2M\sqrt{N} \sum_{k=1}^K\frac{ {\pi}
 \varrho_{r_k} \varrho_d} {4} \mathbb E\left[
|\rho_k|\right]  \nonumber \\
&+  NM  \left(2-\frac{\pi}{2} \right)
\sum_{k=1}^K  \mathbb{E}[|\rho_k|^2]\frac{\varrho_{r_k}^2}{2}+ N\varrho_d^2 \label{pro2}
\end{align}
\label{proposition2}
\end{proposition}
\begin{proof}
See Appendix \ref{appB}.
\end{proof}

We see that, similar to the single IRS case, the average received
signal power attained by the near-optimal analytical solution
scales quadratically with the number of reflecting elements $M$.
Also, as expected, the average received signal power is a sum of
the received signal power from multiple IRSs, which indicates that
better performance can be achieved by deploying multiple IRSs.

\section{Extension to Discrete Phase Shifts} \label{sec:low-resolution-PS}
In previous sections, to simplify our problem, we assume that
elements of IRSs have an infinite phase resolution. Nevertheless,
due to hardware limitations, the phase shift may not take an
arbitrary value, instead, it may have to be chosen from a finite
set of discrete values \cite{TanSun18,WuZhang19d}. Specifically,
the set of discrete values for the phase shift is defined as
\begin{align}
\theta_{k,m} \in {\mathcal{F}} \triangleq \left \{ 0,
\frac{2\pi}{2^b} ,\ldots, \frac{2\pi (2^b-1)}{2^b} \right \}
\end{align}
where $b$ denotes the resolution of the phase shifter. To meet the
finite resolution constraint imposed on the phase shifters, a
simple yet effective solution is to let each phase shift,
$\theta_{k,m}$, take on a discrete value that is closest to its
optimal (or near-optimal) value $\theta_{k,m}^{\star}$ obtained in
previous sections, i.e.
\begin{align}
\theta_{k,m}^{\ast} =  \arg\min_{\theta \in \mathcal{F}} \quad
|\theta-\theta_{k,m}^{\star}| \label{discretized-theta}
\end{align}
where $\theta_{k,m}^{\star}$ denotes the $m$th diagonal entry of
$\boldsymbol{\Theta}_k^{\star}$. In the following, we analyze the
impact of the phase discretization on the system performance. Let
\begin{align}
\gamma(b) &= \mathbb E \bigg[ \bigg \| \sum_{k=1}^K
\boldsymbol{h}_{r_k}^H \boldsymbol{\Theta}_k^{\ast}
\boldsymbol{G}_k  + \boldsymbol{h}_d\bigg\|_2^2 \bigg]
\label{gammab}
\end{align}
denote the average received power attained by our solution with
$b$-bit phase shifters, where
$\boldsymbol{\Theta}_k^{\ast}=\text{diag}(\theta_{k,1}^{\ast},\ldots,\theta_{k,M}^{\ast})$
with $\theta_{k,m}^{\ast}$ given by (\ref{discretized-theta}).
Without loss of generality, we assume the transmit signal power
$p=1$. Our main results are summarized as follows.
\begin{proposition}
Assume $\boldsymbol {h}_{r_k}\sim {\cal CN}(0, \varrho_{r_k}^2
I)$, and the BS-$k$th IRS channel is characterized by (\ref{Gk}).
As $M \rightarrow \infty$, we have
\begin{align}
\eta(b) \triangleq \frac{\gamma(b)}{\gamma(\infty)} = \left(
\frac{2^b}{\pi} \sin\left( \frac{\pi}{2^b}\right)\right)^2
\label{eta}
\end{align} \label{proposition3}
It is not difficult to verify that $\eta(b)$ increases
monotonically with $b$ and approaches $1$ as $b \rightarrow
\infty$.
\end{proposition}
\begin{proof}
See Appendix \ref{appC}.
\end{proof}

This proposition provides a quantitative analysis of the average
received signal power in the multiple-IRS assisted system with
discrete phases shifts. We see that, when compared with the
receive power achieved by IRSs with infinite-resolution phase
shifters, the receive signal power attained by our proposed
solution decreases by a constant factor that depends on the number
of quantization levels $b$. Specifically, we have $\eta(1) =
0.4053$, $\eta(2) = 0.8106$ and $\eta(3) = 0.9496$. Note that a
similar result was also reported in \cite{WuZhang19d}.
Nevertheless, the result in \cite{WuZhang19d} is derived by
considering a simple scenario where both the transmitter and the
receiver are equipped with a single antenna. The extension of the
result in \cite{WuZhang19d} to the multiple transmit antenna
scenario is not straightforward. Also, in \cite{WuZhang19d}, only
a single IRS is employed, whereas our work considers a more
general case where multiple IRS are deployed to assist the
downlink communication.

\section{Simulation Results} \label{sec:simulation-results}
We now present simulation results to illustrate the performance of
the proposed IRS-assisted precoding solutions. In our simulations,
we consider a scenario where the BS employs a ULA with $N$
antennas, and each IRS consists of a uniform rectangular array
(URA) with $M = M_y M_z$ reflecting elements, in which $M_y$ and
$M_z$ denote the number of elements along the horizontal axis and
vertical axis, respectively. The BS-user channel is generated
according to the following geometric channel model
\cite{AyachRajagopal14}:
\begin{align}
\boldsymbol{h}_{d}& =
\sqrt{\frac{N}{L_d}}\sum_{l=1}^{L_d}\alpha_{l} \boldsymbol{a}_{t}(
\phi_{l}) \label{channel-model}
\end{align}
where $L_d$ is the number of paths, $\alpha_l$ is the complex gain
associated with the $l$th path, $\phi_l$ is the associated angle
of departure, $\boldsymbol{a}_t \in \mathbb C^{N}$ represents the
normalized transmit array response vector. The complex gain
$\alpha_l$ is generated according to a complex Gaussian
distribution \cite{AkdenizLiu14}
\begin{align}
\alpha_l\sim {\cal CN}(0,10^{-0.1\kappa}) \label{eqn9}
\end{align}
with $\kappa$ given as
\begin{align}
\kappa= a+10b \log_{10}(\tilde{d}) + \xi
\end{align}
in which $\tilde{d}$ denotes the distance between the transmitter
and the receiver, and $\xi\sim \mathcal{N}(0,\sigma_{\xi}^2)$. The
values of $a$, $b$ $\sigma_{\xi}$ are set to be $a=72$, $b=2.92$,
and $\sigma_{\xi}=8.7$dB, as suggested by real-world NLOS channel
measurements \cite{AkdenizLiu14}.

The IRS-user channel and the BS-IRS channel are generated
according the aforementioned geometric SV model in LOS scenarios.
Specifically, the IRS-user channel is denoted by
\begin{align}
\boldsymbol{h}_{r} &= \sqrt{\frac{M}{L_r}} \left( \varrho_{0}
\boldsymbol{a}_{t} ( \vartheta_{a,0},\vartheta_{e,0}) +
\sum_{l=1}^{L_r-1}\varrho_{l} \boldsymbol{a}_{t}(
\vartheta_{a,l},\vartheta_{e,l}) \right)
\end{align}
where $L_r$ is the number of paths, $ \varrho_{0} $ denotes the
complex gain associated with the LOS component, $\varrho_l$ is the
complex gain associated with the $l$th NLOS path,
$\vartheta_{a,l}$ ($\vartheta_{e,l}$) denotes the azimuth
(elevation) angle of departure associated with the IRS-user path,
$\boldsymbol{a}_t \in \mathbb C^{M}$ represents the normalized
transmit array response vector.

On the other hand, the BS-IRS channel is characterized by the SV
channel model given as
\begin{align}
\boldsymbol{G} =& \sqrt{\frac{{NM}}{L}} \bigg(\alpha
_{0}\boldsymbol {a}_r(\vartheta_{a},\vartheta_{e})
\boldsymbol{a}_t^H(\phi) \nonumber \\
&+ \sum_{i=1}^{L-1}\alpha_{i}
{a}_r(\vartheta_{a_l},\vartheta_{e_l})
\boldsymbol{a}_t^H(\phi_l)\bigg)
\end{align}
where $ \alpha_{0} $ denotes the complex gain with the LOS
component,  $\vartheta_{a_l}$ ($\vartheta_{e_l}$) denotes the
azimuth (elevation) angle of arrival associated with $l$th NLOS
path, $\phi_l$ is the associated angle of departure,
$\boldsymbol{a}_r \in \mathbb C^{M}$ and $\boldsymbol{a}_t \in
\mathbb C^{N}$ represent the normalized receive and transmit array
response vectors, respectively. The complex gain $\alpha_0$ and
$\varrho_0$ are generated according to (\ref{eqn9}). The values of
$a$, $b$ $\sigma_{\xi}$ are set to be $a=61.4$, $b=2$, and
$\sigma_{\xi}=5.8$dB as suggested by LOS real-world channel
measurements \cite{AkdenizLiu14}. The Rician factor (defined as as
the ratio of the energy in the LOS path to the sum of the energy
in other NLOS paths) is set to be $13.2$dB according to
\cite{Muhi-EldeenIvrissimtzis10}. Also, unless specified
otherwise, we assume $N=64$, $M_y=10$, and $M_z=20$ in our
experiments. Other parameters are set as follows: $p=30$dBm,
$\sigma^2=-90$dBm. The average receive SNR is defined as $\mathbb
E[10 \log_{10} \frac{\gamma}{\sigma^2}]$, where $\gamma$ is the
received signal power. Since the noise power is fixed, the
difference between the average receive SNR and the average
received signal power is a constant. All results are averaged over
$1000$ random channel realizations.

\begin{figure}[!t]
\centering {\includegraphics[width=2.7in]{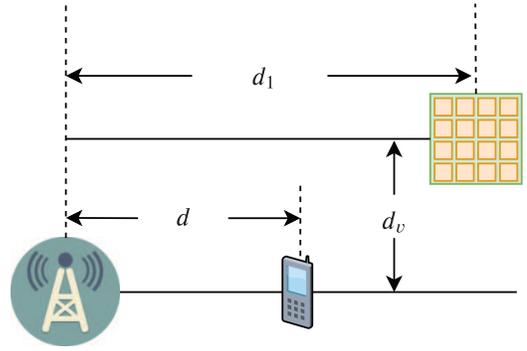}} \caption{
Simulation setup for the single IRS case.} \label{fig2}
\end{figure}

\subsection{Results for Single IRS}
We consider a setup where the IRS lies on a horizontal line which
is in parallel to the line that connects the BS and the user (Fig.
\ref{fig2}). The horizontal distance between the BS and the IRS is
set to $d_1=119$ meters and the vertical distance between two
lines is set to $d_v=0.6$ meters. Let $d$ denote the distance
between the BS and the user. The BS-IRS distance and the IRS-user
distance can then be respectively calculated as
$d_2=\sqrt{d^2+d_v^2}$ and $d_3=\sqrt{(d_1-d)^2+d_v^2}$.

Fig. \ref{fig3} plots the average receive SNRs of our proposed
solutions with both continuous-valued and discrete-valued phase
shifts. Note that for our proposed solutions, the BS-IRS channel
is approximated as a rank-one channel by ignoring those NLOS
paths. The upper bound of the average receive SNR obtained in
\cite{WuZhang18} is included for comparison. Also, to show the
benefits brought by IRSs, a conventional system without IRSs is
considered, where the optimal MRT solution is employed. We see
that our proposed solution with continuous-valued phase shifters
nearly achieves the upper bound of the average receive SNR, which
verifies the optimality of our proposed closed-form solution and
suggests that neglecting the NLOS paths between the BS and the IRS
has little impact on the system performance. Also, with 2-bit
low-resolution phase shifts, our proposed solution can achieve an
average receive SNR close to that attained by assuming
infinite-precision phase shifters. Moreover, it is observed that
for the system without IRSs, the average receive SNR decreases
rapidly as the user moves away from the BS. As a comparison, this
issue can be relieved and the signal coverage can be substantially
enhanced via the use of IRSs.

\begin{figure}[!t]
    \centering
    {\includegraphics[width=3.5in]{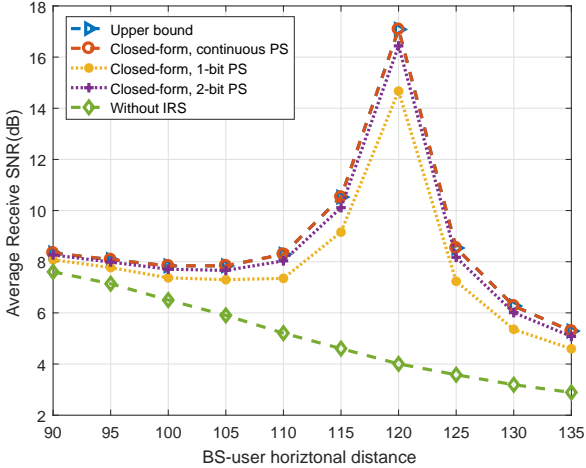}}
    \caption{Average receive SNR versus BS-user horizontal distance, $d$.} \label{fig3}
\end{figure}


In Fig. \ref{fig4}, we plot the average receive SNR versus the number of
reflecting elements at the IRS when $d=119$m, where we fix
$M_y=20$ and increase $M_z$. From Fig. \ref{fig4}, we observe that
the average receive SNR increases quadratically with the number of
reflecting elements. Specifically, the difference between the
receive SNRs when $M=300$ and $M=600$ is approximately equal to
$6$dB, which coincides well with our analysis. In addition, the
average receive SNR loss due to the use of low-resolution phase shifters
is analyzed and given by (\ref{eta}). Specifically, we have
$\eta(1) = -3.9224$dB and $\eta(2)= -0.9121$dB. It can be observed
that simulation results are consistent with our theoretical
result.

\begin{figure}[!t]
    \centering
    {\includegraphics[width=3.5in]{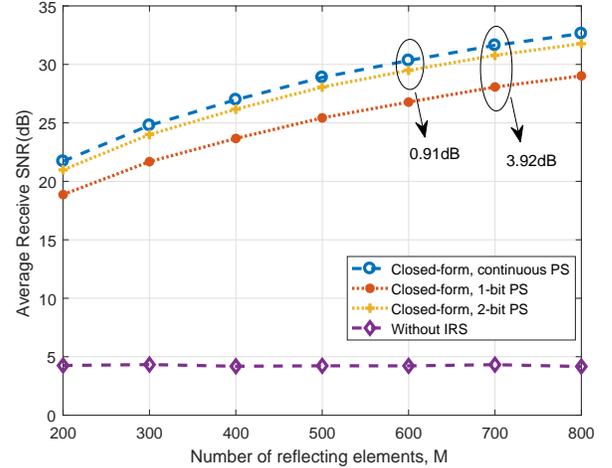}}
    \caption{Average receive SNR versus number of reflecting elements, $M$.} \label{fig4}
\end{figure}

\begin{figure}[!t]
\centering {\includegraphics[width=2.7in]{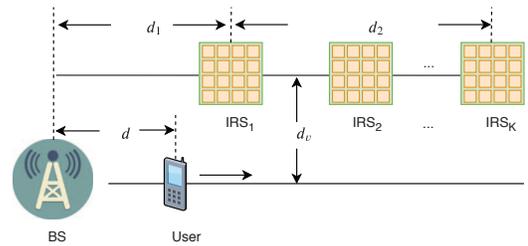}} \caption{
Simulation setup for the multi-IRS case.} \label{fig5}
\end{figure}

\subsection{Results for Multiple IRSs}
We consider a multi-IRS setup as depicted in Fig. \ref{fig5},
where $K$ IRSs are equally spaced on a straight line which is in
parallel with the line connecting the BS and the user.
Specifically, the horizontal distance $d_1$ between the BS and the
first IRS is set to $d_1=100$m and the vertical distance is set to
$d_v=0.6$m. Also, the distance between the nearest IRS and the
farthest IRS is set to be $d_{2}=30$m. We set $K=3$ if not
specified otherwise. In this example, the SV channel model used to
characterize the BS-IRS channel only contains a LOS component.

Fig. \ref{fig6} depicts the average receive SNRs attained by our
proposed SDR-based approach and the near-optimal analytical
solution as a function of the BS-user distance. To verify the
effectiveness of the proposed solutions, an upper bound on the
average receive SNR is obtained by solving the relaxed SDP problem
(\ref{opt11}). We see that the curve of the analytical solution
almost coincides with the upper bound, which validates the
near-optimality of the proposed analytical solution.

\begin{figure}[!t]
    \centering
    {\includegraphics[width=3.5in]{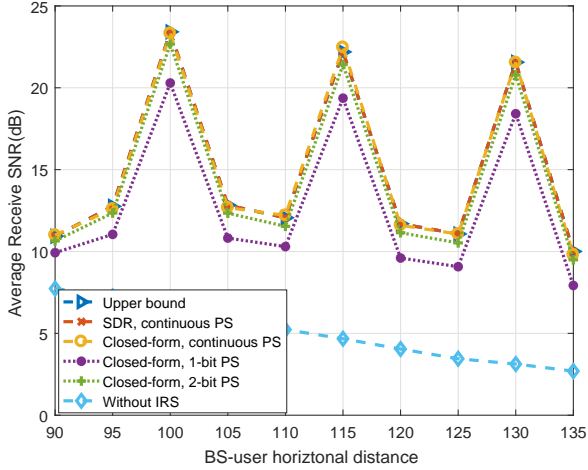}}
    \caption{Average receive SNR versus BS-user horizontal distance for multiple IRSs.} \label{fig6}
\end{figure}

In Fig. \ref{fig7}, we plot the average receive SNRs of different
schemes versus the number of reflecting elements at each IRS,
where we fix $M_y =20$ and change $M_z$. It can be observed that
the squared improvement also holds true for the near optimal
analytical solution. Specifically, when $M=300$, the receive SNR
at the user is approximate to $25$dB, while it increases up to
$31$ dB when the number of reflecting elements doubles, i.e.
$M=600$. Also, we see that the near-optimal analytical solution
achieves an average receive SNR that is closer to the upper bound when $N$
becomes larger, which corroborates our claim that our proposed
analytical solution is asymptotically optimal when $N$ approaches
infinity. Also, it can be seen that the receive SNR loss due to
discretization coincides well with our analysis.

\begin{figure*}[!t]
 \subfigure[Average receive SNR versus number of reflecting elements, $N= 64$.]
 {\includegraphics[width=3.5in]{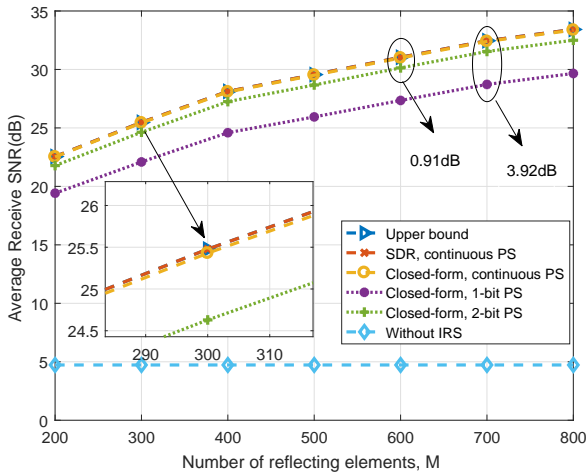}} \hfil
\subfigure[Average receive SNR versus number of reflecting elements, $N=
128$.]{\includegraphics[width=3.5in]{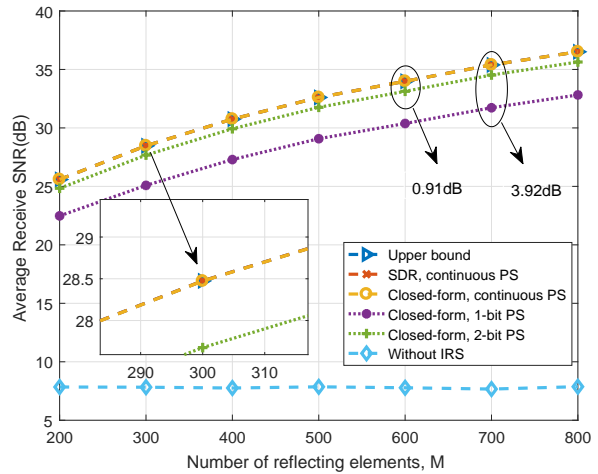}}
\captionsetup{justification=centering}
\caption{Average receive SNR versus number of reflecting elements, $M$}
\label{fig7}
\end{figure*}

To show the robustness of the IRS-assisted system against
blockages, we calculate the average throughput and the outage
probability for our proposed near-optimal analytical solution. The
average throughput $R_a$ and the outage probability are
respectively defined as
\begin{align}
R_a \triangleq\mathbb E \bigg[  \log_{2}
\left(1+\frac{\gamma}{\sigma^2} \right) \bigg]
\end{align}
\begin{align}
\mathbb{P}_{\text{out}}(\tau)= {\mathbb P}(R_a<\tau)
\end{align}
where $\tau$ denotes the required threshold level and set to
$\tau=0.5$ according to \cite{YangDu15}. We assume that the BS-IRS
link is always connected. Also, the blockage probabilities of the
BS-user link and the IRS-user link are assumed to be the same in
our simulations. From Fig. \ref{fig8}(b), we observe that the
outage probability can be substantially reduced by deploying IRSs.
Also, the more the IRSs are deployed, the lower the outage
probability can be achieved. Particularly, when $K=4$, the outage
probability reduces to zero if the link blockage probability is
less than $P<0.1$. This result shows the effectiveness of IRSs in
overcoming the blockage issue that prevents the wider applications
of mmWave communications.

\begin{figure*}[!t]
 \subfigure[Average throughput vs. the link blockage probability.]
 {\includegraphics[width=3.5in]{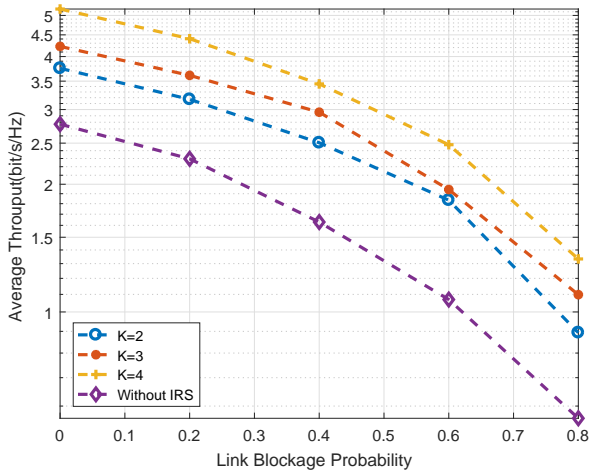}} \hfil
\subfigure[Outage probability vs. the link blockage
probability.]{\includegraphics[width=3.5in]{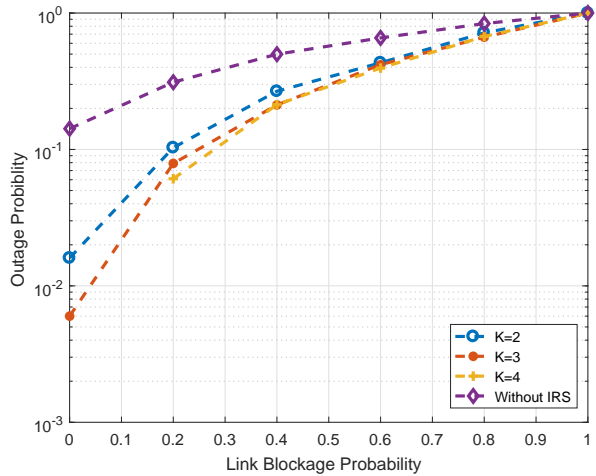}}
\caption{Average throughput and outage probability versus the link
blockage probability, $P$.} \label{fig8}
\end{figure*}

\section{Conclusions} \label{sec:conclusions}
In this paper, we studied the problem of joint active and passive
precoding design for IRS-assisted mmWave systems, where multiple
IRSs are deployed to assist the data transmission from the BS to a
single antenna user. The objective is to maximize the received
signal power by jointly optimizing the transmit precoding vector
at the BS and the phase shift parameters user by IRSs for passive
beamforming. By exploiting some important characteristics of
mmWave channels, we derived a closed-form solution for the single
IRS case, and a near-optimal analytical solution for the multi-IRS
case. Simulation results were provided to illustrate the
optimality and near-optimality of proposed solutions. Our results
also showed that IRSs can help create effective virtual LOS paths
to improve robustness of mmWave systems against blockages.

\useRomanappendicesfalse
\appendices

\section{Proof of Proposition \ref{proposition1}} \label{appA}
When the optimal active and passive beamforming solution is
employed, from (\ref{opt5}), we know that the received signal
power at the user is given as
\begin{align}
\| e^{j \alpha^{\star}} \boldsymbol{h}_r^H
\boldsymbol{\bar{\Theta}}^{\star} \boldsymbol{G}
+\boldsymbol{h}_d^H\|^2_2
&= \|z e^{j \alpha^{\star}}  \boldsymbol{b}^T + \boldsymbol{h}_d^H\|^2_2 \nonumber \\
&= z^2 +  2|z|| \boldsymbol{b}^T\boldsymbol{h}_d| +
\boldsymbol{h}_d^H\boldsymbol{h}_d \label{Pr:Single1}
\end{align}
where
\begin{align}
z \triangleq \sqrt{NM}\rho\boldsymbol{h}_r^H
\boldsymbol{\bar{\Theta}}^{\star}\boldsymbol{a}  =
\sqrt{N}|\rho|\cdot\|\boldsymbol{h}_r\|_{1}
\end{align}
in which the latter equality comes from the fact that
$\rho\boldsymbol{h}_r^H
\boldsymbol{\bar{\Theta}}^{\star}\boldsymbol{a}=\|\rho(\boldsymbol{h}_r^{\ast}\circ\boldsymbol{a})\|_1
=\frac{1}{\sqrt{M}}|\rho|\|\boldsymbol{h}_r\|_1$. Therefore we
have
\begin{align}
\gamma^{\star}=\mathbb{E}[z^2 +  2|z||
\boldsymbol{b}^T\boldsymbol{h}_d| +
\boldsymbol{h}_d^H\boldsymbol{h}_d]
\end{align}
We first calculate $\mathbb{E}[z]$. Since $\boldsymbol{h}_{r}\sim
{\cal CN}(0,\varrho_{r}^2\boldsymbol{I})$, the mean and variance
of the modulus of $m$th entry of $\boldsymbol{h}_{r}$ are
respectively given as
\begin{align}
 \mathbb E[|h_{r_{m}}|] =& \frac{\sqrt{\pi}  \varrho_{r} } {2} \\
\text{Var}\left[|h_{r_{m}}| \right]  =&  \left(2-\frac{\pi}{2}
\right) \frac{\varrho_r^2}{2}
\end{align}
Thus
\begin{align}
\mathbb E\left[  |h_{r_{m}}| ^2\right] = \text{Var}
\left[|h_{r_{m}}| \right]  +  \left(  \mathbb E \left[|h_{r_{m}}|
\right]  \right)^2 = \varrho_r^2
\end{align}
Hence $\mathbb{E}[z]$ can be computed as
\begin{align}
\mathbb  E[z] &= \sqrt{N}\mathbb E[|\rho|]\sum_{m=1}^M \mathbb
E[|h_{r_m}|]\nonumber\\ &= M\sqrt{N}\mathbb E[|\rho|]
\frac{\sqrt{\pi}\varrho_{r}}{2}
\end{align}
and
\begin{align}
& \mathbb E \bigg[ \bigg(\sum_{m=1}^M{|h_{r_{m}}}| \bigg)^2\bigg]  \nonumber \\
= &
\mathbb E \bigg[ \sum_{m=1}^M  | h_{r_{m}} | ^2 + \sum_{i=1}^M
\sum_{j \neq i}^M | h_{r_{i}} || h_{r_{j}} |  \bigg] \nonumber \\
= &   \sum_{m=1}^M  \mathbb E \left[| h_{r_{m}} | ^2   \right]+
\sum_{i=1}^M \sum_{j \neq i}^M  \mathbb E \left[ | h_{r_{i}} | \right]
\mathbb E \left[ | h_{r_{j}} | \right] \nonumber\\
= &  M^2 \frac{\pi \varrho_{r}^2}{4} + M\left(2-\frac{\pi}{2}
\right)\frac{\varrho_{r}^2}{2}
\end{align}
Therefore $\mathbb{E}[z^2]$ is given as
\begin{align}
\mathbb{E}[z^2] &= {N} \mathbb {E}[|\rho| ^2] \mathbb{E}[\|\boldsymbol{h}_r\|_{1}^2]\nonumber\\
& = {N}\mathbb E[|\rho|^2] \mathbb{E}\bigg[ \bigg(\sum_{m=1}^{M}|h_{r_{m}}| \bigg)^2 \bigg]\nonumber\\
& = NM^2 \frac{{\pi}  \varrho_{r}^2} {4} \mathbb{E}[|\rho|^2] + N
M\left(2-\frac{\pi}{2} \right) \mathbb E[|\rho|^2]
\frac{\varrho_{r}^2}{2} \label{eqn2}
\end{align}
Now let us examine $\mathbb{E}[\boldsymbol{b}^T\boldsymbol{h}_d]$.
It is clear that
\begin{align}
     \boldsymbol{b}^T\boldsymbol{h}_d \sim {\cal CN}(0, \varrho_d^2)
\end{align}
As a result, we have
\begin{align}
    \mathbb{E}[|\boldsymbol{b}^T\boldsymbol{h}_d|]=\frac{\sqrt{\pi}}{2} \varrho_d
\end{align}
and
\begin{align}
    \mathbb {E} (|z|| \boldsymbol{b}^T\boldsymbol{h}_d|) =
    M \sqrt{N}\mathbb E(|\rho|)  \frac{{\pi}  \varrho_r \varrho_d}
    {4} \label{eqn3}
\end{align}
In addition, it can be easily verified that
\begin{align}
    \mathbb E[ \pmb h_d^H\pmb h_d ] = \sum_{n=1}^{N} \mathbb E[
    |h_{d_n}|^2]=N\varrho_d^2 \label{eqn4}
\end{align}
Combining (\ref{eqn2}), (\ref{eqn3}) and (\ref{eqn4}), we reach
(\ref{power-scaling-law}). This completes our proof.

\section{Proof of Proposition \ref{proposition2}} \label{appB}
When the analytical active and passive beamforming solution is
employed, from \eqref{opt12}, we know that the received signal
power at the user is given as
\begin{align}
&\| (\boldsymbol{v}^{\star})^H \boldsymbol{ \Phi} + \boldsymbol{h_d}^H\|_2^2  \nonumber \\
=& {(\boldsymbol{v}^{\star}})^H \boldsymbol{\Phi \Phi}^H
{\boldsymbol{v}^{\star}}+(\boldsymbol{v}^{\star})^H
\boldsymbol{\Phi} \boldsymbol h_d + \boldsymbol
h_d^H\boldsymbol{\Phi} ^H {\boldsymbol{v}^{\star}}
+ \boldsymbol{h}_d ^H \boldsymbol{h}_d \nonumber \\
\stackrel{(a)}{\approx}  & \| \boldsymbol z\|_2^2 + 2
|(\boldsymbol{v}^{\star})^H \boldsymbol{\Phi} \boldsymbol h_d |
 + \boldsymbol{h}_d ^H \boldsymbol{h}_d
\label{Pr:Multi1}
\end{align}
where $\boldsymbol{v}^{\star}$ is given by (\ref{v-opt}), $(a)$ is
due to (\ref{eqn8}), and
\begin{align}
z_k =\sqrt{N} |\rho_k| \| \boldsymbol{h}_{r_k}\|_{1}
\label{def:z2}
\end{align}
Therefore we have
\begin{align}
\gamma \approx\mathbb E \left[ \| \boldsymbol z\|_2^2 +2
|(\boldsymbol{v}^{\star})^H \boldsymbol{\Phi} \boldsymbol h_d | +
\boldsymbol{h}_d ^H \boldsymbol{h}_d  \right]
\end{align}
We first calculate $\mathbb{E}[|z_k|]$. Since $\boldsymbol
{h}_{r_k}\sim {\cal CN}(0, \varrho_{r_k}^2 I)$, we have
\begin{align}
 \mathbb E \left[|h_{r_{k,m}}|\right] = \frac{\sqrt{\pi}  \varrho_{r_k} } {2}
\end{align}
\begin{align}
\text{Var} \left[|h_{r_{k,m}}|\right] = \left(2- \frac{\pi}{2}
\right) \frac{\varrho_{r_{k}}^2 }{2}
\end{align}
\begin{align}
 \mathbb E \left[| h_{r_{k,m}} | ^2 \right] = \text{Var}\left[|h_{r_{k,m}}|\right]  +
 \left( \mathbb E \left[|h_{r_{k,m}}|\right] \right)^2 = \varrho_{r_{k}}^2
\end{align}
Hence $\mathbb E[|z_k|]$ can be computed as
\begin{align}
\mathbb  E \left[|z_k|\right] &= \sqrt{N}\mathbb E
\left[|\rho_k|\right] \sum_{m=1}^M \mathbb E
\left[|h_{r_{k,m}}|\right]\nonumber\\
&= M\sqrt{N}\mathbb E\left[|\rho_k|\right]  \frac{\sqrt{\pi}
\varrho_{r_k}} {2}
\end{align}
and
\begin{align}
\mathbb E \bigg[ \bigg(\sum_{m=1}^M{|h_{r_{k,m}}}| \bigg)^2\bigg]
&= \mathbb E \bigg[ \sum_{m=1}^M  | h_{r_{k,m}} | ^2 +\sum_{i=1}^M
\sum_{j \neq i}^M | h_{r_{k,i}} || h_{r_{k,j}} |  \bigg] \nonumber \\
& = M \varrho_{r_{k}}^2 + M(M-1) \frac{\pi \varrho^2_{r_{k}} }{4}    \nonumber \\
& = M^2 \frac{\pi \varrho_{r_k}^2}{4} + M\left(2-\frac{\pi}{2}
\right)\frac{\varrho_{r_k}^2}{2}
\end{align}
Therefore, we have
\begin{align}
\mathbb{E}[|z_k|^2] & =
{N}\mathbb E[|\rho_k|^2] \mathbb{E}\bigg[ \bigg(\sum_{m=1}^{M}|h_{r_{k,m}}| \bigg)^2 \bigg]\nonumber\\
& = NM^2 \mathbb{E}[|\rho_k|^2]  \frac{{\pi} \varrho_{r_k}^2}{4} +
NM\mathbb{E}[|\rho_k|^2] \left(2-\frac{\pi}{2}
\right)\frac{\varrho_{r_k}^2}{2}
\end{align}
and
\begin{align}
    \mathbb{E} \left[\|\boldsymbol{z}\|_2^2\right]
    =& \mathbb{E} \left[\sum_{k=1}^K |z_k|^2\right] \nonumber \\
     =&  NM^2  \sum_{k=1}^K E \left[|\rho_k|^2 \right] \frac{{\pi}  \varrho_{r_k}^2} {4} \nonumber \\
  &  +  NM  \left(2-\frac{\pi}{2} \right) \sum_{k=1}^K  \mathbb{E}[|\rho_k|^2]\frac{\varrho_{r_k}^2}{2}
    \label{eqn5}
\end{align}

Now we examine $\mathbb E \left[ |({\boldsymbol{v}^{\star}})^H
\boldsymbol{\Phi} \boldsymbol h_d |\right] $. From (\ref{v-opt}),
we arrive

\begin{align}
\mathbb E \left[ |({\boldsymbol{v}^{\star}})^H \boldsymbol{\Phi}
\boldsymbol h_d |\right] = \mathbb E \left[
|({\boldsymbol{v}^{\star}})^H \boldsymbol u \right] = \mathbb E
\left[  \sum_{k=1}^K |u_k|  \right] = \sum_{k=1}^K \mathbb E
\left[ |u_k| \right]
\end{align}
we can verify that
\begin{align}
\boldsymbol{u} = \boldsymbol{\Phi} \boldsymbol{h}_d \sim {\mathcal
{CN}}(0,\varrho_d^2 {\text{diag}}(z_1^2,\ldots,z_K^2))
\end{align}
Since $z_k$  is also a random variable, we can calculate $\mathbb
E[|u_k|]$ as
\begin{align}
\mathbb E[|u_k|] &= \mathbb E[ \mathbb E [|u_k| |z_k] ] \nonumber \\
&= \frac{\sqrt{\pi} \varrho_d}{2} \mathbb E[z_k] \nonumber \\
& =  M \sqrt{N}\frac{\pi \varrho_d \varrho_{r_k} }{4} \mathbb
E[|\rho_k|]
\end{align}
As a result, we have
\begin{align}
\mathbb E \left[ |({\boldsymbol{v}^{\star}})^H \boldsymbol{\Phi}
\boldsymbol h_d |\right] = M \sqrt{N} \sum_{k=1}^K \frac{\pi
\varrho_d \varrho_{r_k} }{4} \mathbb E[|\rho_k|] \label{eqn6}
\end{align}

Additionally, it can be easily verified that
\begin{align}
    \mathbb E \left[ \pmb h_d^H\pmb h_d \right] = \sum_{n=1}^{N}
    \mathbb E \left[|h_{d_n}|^2 \right]= N\varrho_d^2
     \label{eqn7}
\end{align}
Combining (\ref{eqn5}), (\ref {eqn6}) and (\ref{eqn7}), we reach
(\ref{pro2}). This completes our proof.

\section{Proof of Proposition \ref{proposition3}} \label{appC}
To facilitate our analysis, we rewrite \eqref{discretized-theta}
as
\begin{align}
\boldsymbol{\Theta}_k^{\ast} = e^{j \alpha_k^{\star}}
\boldsymbol{ \bar \Theta} ^{\star} _k\Delta\boldsymbol{\Theta}_k
\triangleq  e^{j \alpha_k^{\star}}   \boldsymbol{ \tilde \Theta}_k
\label{opt-dis}
\end{align}
where $ \boldsymbol{ \tilde \Theta}_k \triangleq \boldsymbol{ \bar
\Theta} ^{\star} _k\Delta\boldsymbol{\Theta}_k  $,
$\Delta\boldsymbol{\Theta}_k \triangleq {\text {diag}} (e^{j\Delta
\theta_{k,m}}, \ldots, e^{j\Delta \theta_{k,M} } )$, in which
$\Delta \theta_{k,m}$ is the discretization error. By substituting
(\ref{opt-dis}) into (\ref{gammab}), the average received signal
power is given by
\begin{align}
\gamma(b)    = & \mathbb E \bigg[ \bigg\| \sum_{k=1}^K
\sqrt{NM}\rho_k\boldsymbol{h}_{r_k}^H
\boldsymbol{\Theta}_k^{\ast}\boldsymbol{a}_k\boldsymbol{b}_k^T +
\boldsymbol{h}_d\bigg\|_2^2 \bigg]          \nonumber \\
 =&  \mathbb E \bigg[ \bigg\| \sum_{k=1}^K  e^{j\alpha_k^{\star}}
\; \sqrt{NM}\rho_k \boldsymbol{h}_{r_k}^H \boldsymbol{\tilde
\Theta}_k\boldsymbol{a}_k\boldsymbol{b}_k^T +
\boldsymbol{h}_d\bigg\|_2^2 \bigg]
\nonumber \\
\stackrel{(a)}  = & \mathbb E \bigg[ \bigg\| \sum_{k=1}^K
e^{j\alpha_k^{\star}} \; {\tilde z}_k \boldsymbol{b}_k^T +
\boldsymbol{h}_d\bigg\|_2^2 \bigg]  \nonumber \\
\stackrel{(b)} =  &\mathbb E \bigg[ \bigg \|
(\boldsymbol{v^{\star}})^H
\boldsymbol{\tilde D}_z \boldsymbol{B}  + \boldsymbol{h}_d   \bigg\|_2^2  \bigg] \nonumber \\
\stackrel{(c)} \approx & \mathbb E \bigg[  \sum_{k=1}^K
|\tilde{z}_k|^2  + (\boldsymbol{v^{\star}})^H \boldsymbol{\tilde
D}_z \boldsymbol{B} \boldsymbol{h}_d  + \boldsymbol{h}_d ^H
(\boldsymbol{\tilde D}_z \boldsymbol{B})^H
(\boldsymbol{v^{\star}}) \nonumber \\
&+ \boldsymbol{h}_d^H
\boldsymbol{h}_d \bigg] \nonumber \\
=& \mathbb E \left[  \sum_{k=1}^K |\tilde{z}_k|^2\right] +
2\mathbb{R} \bigg\{  \mathbb E \left [  (\boldsymbol{v^{\star}})^H
\boldsymbol{\tilde D}_z \boldsymbol{B} \boldsymbol{h}_d  \right]
\bigg\} +   \mathbb E \left[ \boldsymbol{h}_d^H
\boldsymbol{h}_d\right]
\end{align}
where in $(a)$, we define
\begin{align}
\tilde{z}_k \triangleq \sqrt{NM}\rho_k\boldsymbol{h}_{r_k}^H
\boldsymbol{\tilde \Theta}_k  \boldsymbol{a}_k =
\sqrt{N}|\rho_k|\cdot \sum_{m=1}^M |h_{r_{k,m}}|e^{j\Delta
\theta_{k,m}}
\end{align}
in $(b)$, we define $\boldsymbol{\tilde{D}}_{z} \triangleq {\text
{diag}} (\tilde{z}_1, \ldots, \tilde{z}_K)$ and $(c)$ comes from
\eqref{eqn8}.

Since discrete phase shift values in $\mathcal{F}$ are uniformly
spaced, discretization errors $\{\Delta \theta_{k,m} \}$ can be
considered as independent random variables uniformly distributed
on the interval $[- \frac{\pi}{2^b},\frac{\pi}{2^b}]$. Thus
$\mathbb E[\tilde z_k]$ can be calculated as
\begin{align}
\mathbb E[\tilde z_k] & =  \mathbb E \left[ \sqrt{N}|\rho_k|\cdot
\sum_{m=1}^M |h_{r_{k,m}}|e^{j\Delta
\theta_{k,m}} \right] \nonumber \\
& = \sqrt{N} \mathbb E[ |\rho_k| ] \sum_{m=1}^M \mathbb E[
|h_{r_{k,m}}| ] \mathbb E[ e^{j\Delta
\theta_{k,m}} ] \nonumber \\
&=  M\sqrt{N} \mathbb E[ |\rho_k| ]
\frac{\sqrt{\pi}\varrho_{r_k}}{2} \frac{2^b}{\pi} \sin{\left(
\frac{\pi}{2^b} \right)} \label{zk-dis}
\end{align}
in which
\begin{align}
\mathbb E [ e^{j\Delta \theta_{k,m}} ] = \mathbb E [ -e^{j\Delta
\theta_{k,m} }] = \frac{2^b}{\pi} \sin{\left( \frac{\pi}{2^b}
\right)}
\end{align}
Also, $\mathbb{E}[|\tilde{z}_k|^2] $ can be computed as
\begin{align}
\mathbb E [|\tilde{z}_k|^2] = &  N \mathbb E[|\rho_k|^2]  \mathbb
E\bigg[ \sum_{m=1}^M
|h_{r_{k,m}}|^2 \bigg] \nonumber \\
&  +N \mathbb E[|\rho_k|^2] \mathbb E \bigg[\sum_{m=1}^M \sum_{i
\neq m}^M |h_{r_{k,m}}| |h_{r_{k,i}}|  e^{j(\Delta \theta_{k,m}-
\Delta \theta_{k,i})} \bigg] \nonumber \\
 = &  NM  \varrho_{r_k}^2 \mathbb E[|\rho_k|^2]  \nonumber \\
&  +  N M(M-1)  \mathbb E[|\rho_k|^2]  \frac{{\pi}
\varrho_{r_k}^2} {4} \left(\frac{2^b}{\pi} \sin{\left(
\frac{\pi}{2^b} \right)} \right)^2 \label{zk2-dis}
\end{align}

Next, from (\ref{v-opt}), we arrive at
\begin{align}
  & \mathbb{R} \bigg\{  \mathbb E \left [  (\boldsymbol{v^{\star}})^H
\boldsymbol{\tilde D}_z \boldsymbol{B} \boldsymbol{h}_d  \right]
\bigg\}
 \stackrel{(a)}=   \mathbb{R}\bigg\{ \mathbb E \left[ \sum_{k=1}^K|s_k| \tilde z_k  \right]  \bigg \} \nonumber \\
=&  \sum_{k=1}^K \mathbb E [|s_k|] \mathbb{R} \bigg \{  \mathbb E[
\tilde z_k   ]
\bigg \} \nonumber \\
 = & \sum_{k=1}^K \frac{\sqrt{\pi} \varrho_d}{2} \mathbb E[  \tilde z_k   ] \nonumber \\
 =& \sum_{k=1}^K  M\sqrt{N}   \frac{\sqrt{\pi} \varrho_d}{2} \mathbb E[ |\rho_k| ]
 \frac{\sqrt{\pi}\varrho_{r_k}}{2} \frac{2^b}{\pi} \sin{\left( \frac{\pi}{2^b} \right)} \nonumber \\
=&  \frac{2^b}{\pi} \sin{\left( \frac{\pi}{2^b} \right)}
M\sqrt{N} \sum_{k=1}^K  \frac{\pi \varrho_d \varrho_{r_k} }{4}
\mathbb E[ |\rho_k| ] \label{mid-dis}
\end{align}
where in $(a)$, $s_k$ is the $k$th entry of $\boldsymbol{s}$ and
$\boldsymbol{s} \triangleq \boldsymbol{B} \boldsymbol{h}_d \sim
\mathcal{CN}(0,\varrho_d^2 \boldsymbol{I})$,  $(a)$ is due to the
fact that $\alpha_k^{\star} = -{\text {arg}} (u_k) = -{\text
{arg}} (s_k)$.

Combining \eqref{zk2-dis}, \eqref{mid-dis} and \eqref{eqn7}, the
average received power can be given as
\begin{align}
\gamma(b) = & NM \sum_{k=1}^K  \varrho_{r_k}^2 \mathbb E[|\rho_k|^2]  \nonumber \\
&  +  N M(M-1) \sum_{k=1}^K \mathbb E[|\rho_k|^2]  \frac{{\pi}
\varrho_{r_k}^2} {4} \left(\frac{2^b}{\pi} \sin{\left(
\frac{\pi}{2^b} \right)} \right)^2
\nonumber \\
& + 2 \frac{2^b}{\pi} \sin{\left( \frac{\pi}{2^b} \right)}
M\sqrt{N} \sum_{k=1}^K  \frac{\pi \varrho_d \varrho_{r_k} }{4}
\mathbb E[ |\rho_k| ]  + N\varrho_d^2
\end{align}
It can be easily obtained that the ratio of $\gamma(b)$ to
$\gamma(\infty)$ is given by (\ref{eta}) as $M \rightarrow
\infty$. This completes our proof.

\bibliography{newbib}
\bibliographystyle{IEEEtran}


\end{document}